\documentclass{ieeeaccess}

\usepackage{cite}
\usepackage{amsmath,amssymb,amsfonts,amsthm}
\usepackage{graphicx}
\usepackage{booktabs}
\usepackage{array}
\usepackage{textcomp}
\usepackage{hyperref}

\graphicspath{{figures/}}

\newcommand{\ket}[1]{|#1\rangle}

\newcommand{\braopket}[3]{\langle#1|#2|#3\rangle}
\newcommand{\rgrad}{\mathrm{rgrad}}
\newcommand{\sym}{\mathrm{sym}}
\newcommand{\polinaorcid}{\href{https://orcid.org/0000-0003-3288-3111}{0000-0003-3288-3111}}

\newtheorem{proposition}{Proposition}

\history{}
\doi{}

\title{A Controlled Null Result for Per-Gate Riemannian Optimization in VQE}
\author{\uppercase{Polina Belozerova}\authorrefmark{1} and
\uppercase{Pavel Osinenko}\authorrefmark{1},\authorrefmark{2},\authorrefmark{3}}

\address[1]{AI in Dynamic Action Lab, Central University, 7 Gasheka St.,
Bldg. 1, Moscow, 123056, Russia (e-mail:
p.belozerova@centraluniversity.ru; p.osinenko@gmail.com)}
\address[2]{Center for Engineering Systems and Sciences, Skoltech,
30 Bolshoy Boulevard, Bldg. 1, Moscow, 121205, Russia}
\address[3]{Sirius University of Science and Technology, 1 Olimpiyskiy
Avenue, Sirius, 354340, Russia}

\corresp{Corresponding author: Polina Belozerova
(ORCID: \polinaorcid; e-mail: p.belozerova@centraluniversity.ru).}

\begin{document}

\begin{abstract}
Parametrized gates in a variational quantum eigensolver (VQE) live on a product
of unitary groups, which makes Riemannian optimization a natural candidate for
improving training.  We test this candidate in the clean per-gate setting: each
single-qubit gate is optimized directly on $U(2)$ and compared with Adam on an
unconstrained Euler-angle chart representing the same physical gate, initialized
from the same unitaries and tuned with a separate learning-rate sweep.  On a
validated hydrogen-chain benchmark, we
find no measurable practical advantage from the Riemannian machinery.  The
direct manifold optimizer and the Euler-coordinate baseline are interleaved in
final error and success rate; the classical coordinate baseline is not slower in
iterations on the successful cells; and the choices of Riemannian metric,
retraction, and momentum vector transport do not change the conclusion.  We
also identify four failure modes that can manufacture an apparent positive
result: a trivial product-state benchmark, a projected-matrix Adam baseline that
discards most of its own step, a non-unitary inherited Cayley retraction that
lets the reported energy pass below the variational minimum, and a learning-rate
grid whose optimum sits at its boundary.  Each failure mode produced a plausible
multi-seed result in our own workflow before being caught.  The practical output
of this paper is therefore a controlled null result and a validation checklist
for VQE optimizer comparisons.
\end{abstract}


\titlepgskip=-21pt
\maketitle

\section{Introduction}

The variational quantum eigensolver minimizes the Rayleigh quotient
$E(\theta)=\braopket{\psi(\theta)}{H}{\psi(\theta)}$ over a parametrized
quantum circuit \cite{peruzzo2014,cerezo2021}.  Since quantum gates are
unitary, the parameter space is not intrinsically Euclidean.  This has motivated
Riemannian-gradient and natural-gradient methods for quantum circuits
\cite{luchnikov2021,stokes2020,wiersema2023}, as well as the hope that direct
optimization over gate matrices might repair training pathologies seen in angle
coordinates.

The present paper asks a narrower question.  In a hardware-efficient VQE ansatz
whose trainable factors are single-qubit $U(2)$ gates, does direct per-gate
Riemannian Adam outperform a matched classical Adam baseline on Euler angles?
This is the cheapest setting in which the conjecture can be tested cleanly: the
two methods represent the same physical circuit, share the same initial
unitaries, and differ only in the coordinate or manifold representation used by
the optimizer.

Our answer is negative in the resolved benchmark cells.  Once the baseline is
physically matched, learning rates are swept per method, runtime validity guards
are enforced, and invalid or under-calibrated cells are excluded, no stable
advantage remains for the Riemannian representation.  This is not a statement
about all Riemannian quantum optimization.  Large-block circuit compilation,
noisy hardware training, quantum natural-gradient preconditioning, and other
architectures are outside the scope of this experiment.  The claim is only that
per-gate $U(2)$ Riemannian machinery does not produce a detectable practical
benefit in the validated exact-state VQE regime studied here.

The paper makes three contributions.

\begin{enumerate}
  \item It reports a controlled null result for per-gate $U(2)$ Riemannian Adam
        against an Euler-angle Adam baseline with matched initialization and
        learning-rate selection.
  \item It shows that one apparent ablation axis is mathematically vacuous: on
        square unitary factors the canonical and Euclidean Riemannian gradients
        differ by exactly a constant factor, which scale-invariant adaptive
        optimizers absorb.
  \item It documents four measured failure modes and distills them into a
        validation protocol for future VQE optimizer comparisons.
\end{enumerate}

The stronger question of why VQE depth thresholds appear, and how to separate
representability floors from trainability bottlenecks, is deliberately reserved
for the companion threshold-mechanism manuscript.

\section{Benchmark and Methods}

\subsection{Validated molecular ladder}

A VQE optimizer benchmark is meaningful only if the target is non-trivial for
the ansatz.  We use neutral hydrogen chains in STO-3G with Jordan--Wigner
mapping and 1.0~\AA{} spacing.  System size is varied by truncating virtual
orbitals, not by freezing occupied orbitals.  The latter can empty the active
space and create product-state ``benchmarks'' that any optimizer can solve.

Table~\ref{tab:bench} records the validation diagnostics used before comparing
optimizers: exact ground energy, particle number, correlation against the
Hartree--Fock state, Schmidt rank, and half-chain entropy.  The 12-qubit entry
is useful for threshold studies but is excluded from the headline optimizer
comparison because no tested depth solves it.

\begin{table*}[t]
\centering\small
\begin{tabular}{rlrrrrrr}
\toprule
qubits & system & $e^-$ & $\lambda_{\min}$ (Ha) & $\langle N\rangle$ &
$\lambda_{\min}-E_{\mathrm{HF}}$ & Schmidt rank & $S_{1/2}$ \\
\midrule
4  & H$_2$              & 2 & $-1.101150$ & 2.000 & $-0.0350$ & 2  & 0.198 \\
6  & H$_4$, 3 orbitals  & 4 & $-2.127938$ & 4.000 & $-0.0294$ & 5  & 0.218 \\
8  & H$_4$              & 4 & $-2.166387$ & 4.000 & $-0.0678$ & 16 & 0.426 \\
10 & H$_6$, 5 orbitals  & 6 & $-3.200749$ & 6.000 & $-0.0652$ & 22 & 0.509 \\
12 & H$_6$              & 6 & $-3.236066$ & 6.000 & $-0.1005$ & 56 & 0.675 \\
\bottomrule
\end{tabular}
\caption{Benchmark validation.  Every listed target is entangled, correlated by
more than chemical accuracy, and in the correct particle-number sector without a
penalty term.  The 12-qubit target is not used to rank optimizers because no
tested depth solves it.}
\label{tab:bench}
\end{table*}

\subsection{Circuit and optimizer comparison}

The ansatz starts from $\ket{+}^{\otimes n}$.  Each layer applies one trainable
$U(2)$ gate per qubit followed by a fixed CNOT entangler.  The simulation stores
the full state vector as a rank-$n$ tensor and applies gates by tensor
contractions; it is not a truncated tensor-train computation.

The manifold method treats every trainable gate as a point of
$\mathrm{St}(2,2)=U(2)$.  Given Euclidean gradient $G$ at $U$, the Euclidean
Riemannian gradient is
\begin{align}
  \rgrad_{\mathrm{euc}} &= G - U\sym(U^\dagger G),
  \label{eq:rgrad-euc}\\
  \sym(A)&=\frac{A+A^\dagger}{2}.
\end{align}
We use a QGOpt-style Riemannian Adam update with a tangent first moment, scalar
second moment, optional vector transport, and a selected retraction
\cite{luchnikov2021,becigneul2019}.  The safe retractions considered in this
study are polar, QR, exact exponential, and an antisymmetrized Cayley map.

The primary baseline is Adam on an unconstrained Euler chart for the same full
$U(2)$ gate,
\begin{equation}
  U = e^{i\phi} R_z(a)R_x(b)R_z(c).
  \label{eq:euler}
\end{equation}
The initial unitaries are identical: the random manifold gates are decomposed
into Euler coordinates and reconstructed before the run.  This baseline never
leaves the unitary group and is the relevant comparison for practical circuit
training.  Adam on raw matrix entries followed by projection is retained only as
a failure-mode study, not as a valid competitor.

Learning rates are swept separately for each method.  If the best value lands
on the smallest or largest grid point, the cell is considered unresolved until
the grid is expanded.  Seeds are paired across methods.  We report both success
rate at chemical accuracy ($1.6\times10^{-3}$ Ha) and geometric mean best
energy error.

\begin{figure*}[t]
\centering
\includegraphics[width=0.78\textwidth]{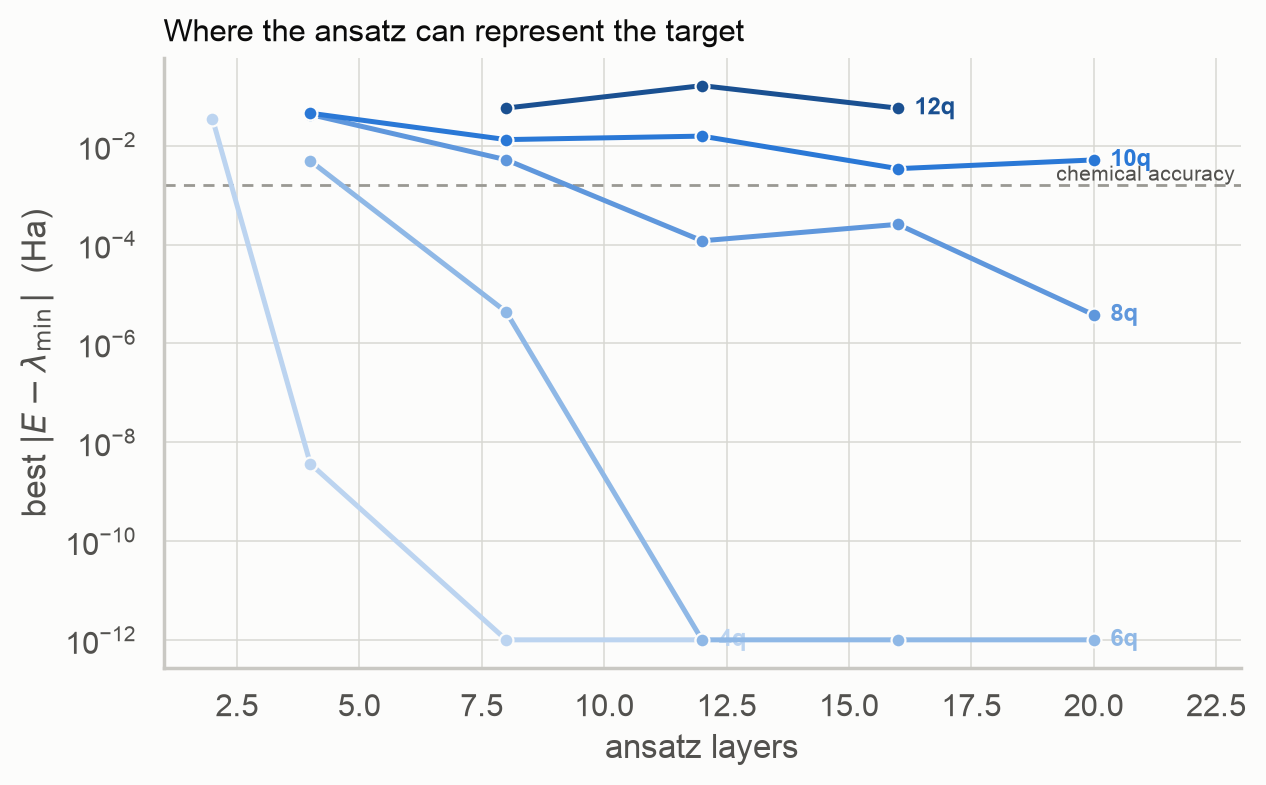}
\caption{Depth calibration used before ranking optimizers.  Operating points are
chosen near the shallowest depths that can reach chemical accuracy; cells that
no tested method can solve are not used as evidence for optimizer superiority.}
\label{fig:depth}
\end{figure*}

\section{Controlled Null Result}

\subsection{Final error and success rate}

Table~\ref{tab:baseline} and Fig.~\ref{fig:baseline} summarize the
headline comparison.  The manifold optimizer is better at two sizes and worse
at two sizes, with no monotone trend in system size.  The projected matrix
baseline is worse by orders of magnitude, but this is an artifact of the
baseline itself, not evidence for a useful geometric advantage.

\begin{table*}[t]
\centering\small
\begin{tabular}{rrlll r}
\toprule
qubits & layers & manifold $U(2)$ & Euler-angle Adam & projected Adam & ratio \\
\midrule
4  & 4  & $7.8\times10^{-8}$ (100\%) & $6.0\times10^{-8}$ (100\%) & $2.1\times10^{-1}$ (0\%) & $0.76$ \\
6  & 8  & $1.3\times10^{-5}$ (100\%) & $2.5\times10^{-5}$ (100\%) & $4.5\times10^{-1}$ (0\%) & $1.91$ \\
8  & 12 & $2.1\times10^{-3}$ (50\%)  & $8.0\times10^{-4}$ (83\%)  & $7.9\times10^{-1}$ (0\%) & $0.37$ \\
10 & 16 & $1.4\times10^{-2}$ (0\%)   & $1.6\times10^{-2}$ (0\%)   & $1.5\times10^{0}$  (0\%) & $1.19$ \\
\bottomrule
\end{tabular}
\caption{Geometric mean best $|E-\lambda_{\min}|$ and success rate at the best
learning rate for each method.  The ratio is Euler/manifold, so values above
one favor the manifold optimizer.  The comparison is interleaved rather than
systematically positive.  The projected-Adam row is shown only as a diagnostic
failure mode.}
\label{tab:baseline}
\end{table*}

\begin{figure*}[t]
\centering
\includegraphics[width=0.82\textwidth]{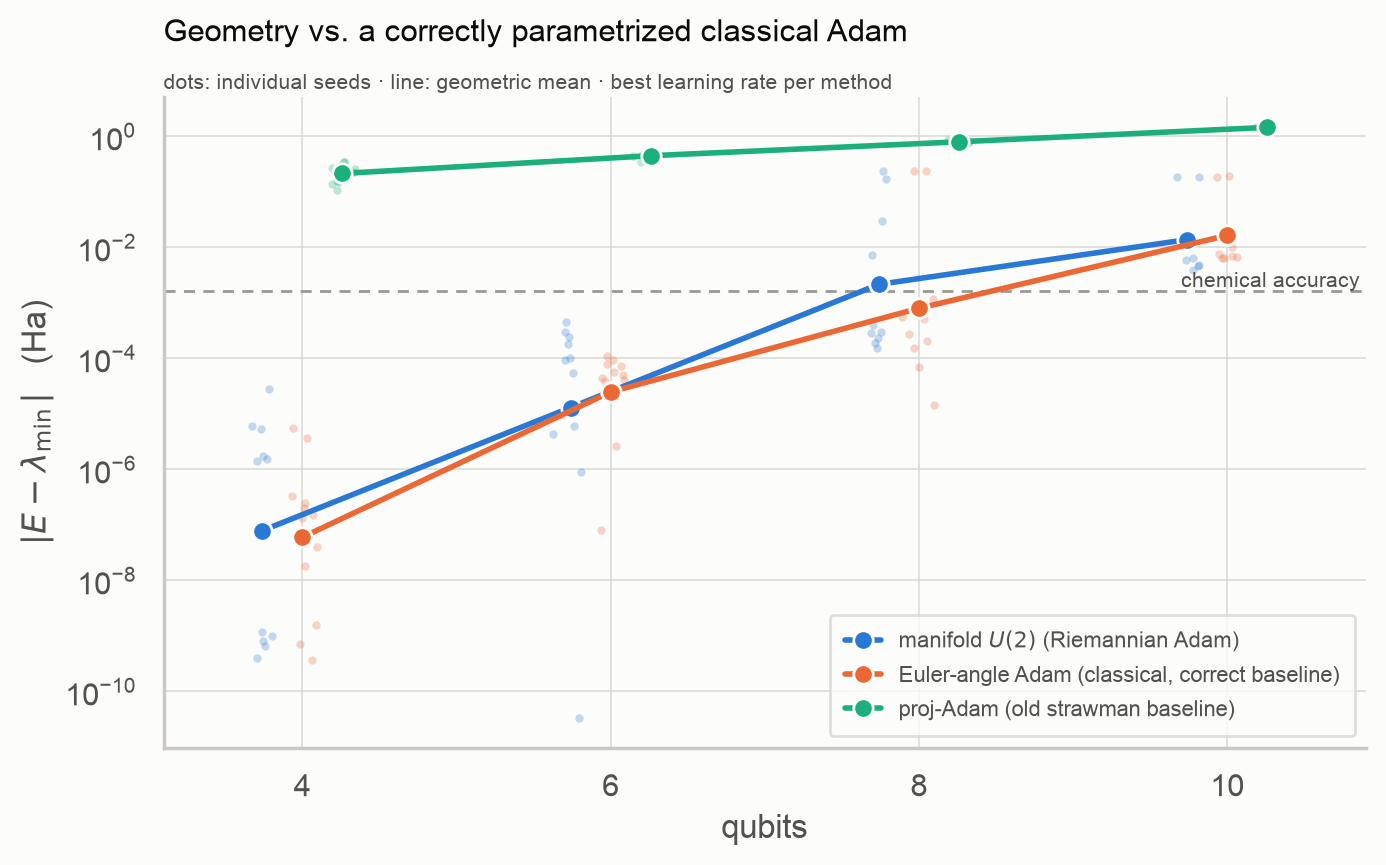}
\caption{Matched Euler-coordinate Adam versus direct per-gate manifold Adam.
Dots denote individual seeds and lines denote the geometric mean at the best
learning rate per method.  The two valid methods are interleaved; projected
matrix Adam sits far above both.}
\label{fig:baseline}
\end{figure*}

\subsection{Iterations to target}

Final error alone can hide a useful speed difference.  On hardware, each
iteration corresponds to another round of circuit evaluations, so iteration
count is a practical cost.  In the successful cells the coordinate
baseline is not slower in iterations and is faster at 6q and 8q.  This supports
the interpretation that the null result is not merely a final-accuracy artifact.
The exact wall-clock timing of the two code paths is implementation-dependent
and is therefore not used as a primary scientific claim.

\subsection{Metric axis}

\begin{proposition}[Metric collapse on square unitary factors]
For $U\in U(n)$, the Euclidean and canonical Riemannian gradients induced by a
Euclidean ambient gradient $G$ obey
\begin{equation}
  \rgrad_{\mathrm{can}} = 2\,\rgrad_{\mathrm{euc}}
\end{equation}
under the conventions used in this implementation.
\end{proposition}

\begin{proof}
On a square unitary factor, $U^\dagger U=I$ and the tangent space is
$\{U\Omega:\Omega^\dagger=-\Omega\}$.  Substituting the Stiefel canonical and
Euclidean projections and simplifying with $U^\dagger U=I$ leaves the same
skew-Hermitian tangent generator up to the constant factor two.  Adam with a
scalar second moment is invariant to such a global rescaling of the gradient,
apart from numerical epsilon effects.
\end{proof}

The experiment confirms this collapse: the canonical-metric row reproduces the
Euclidean row to the printed precision.  Hence the metric choice cannot explain
any positive or negative result in this per-gate $U(2)$ setting.

\subsection{Retraction and transport}

With correct implementations, the safe retractions are indistinguishable at the
level of success counts and final errors.  Momentum vector transport also does
not produce a measurable effect.  Fig.~\ref{fig:ablation} records the compact
ablation panel.

\begin{figure*}[t]
\centering
\includegraphics[width=0.95\textwidth]{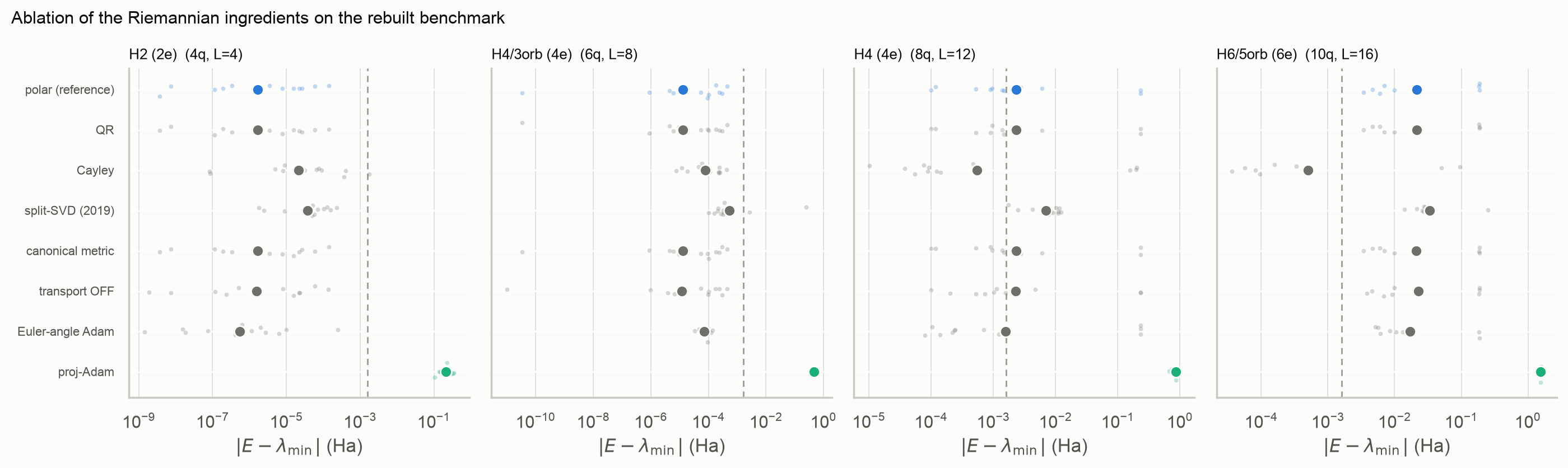}
\caption{Ingredient and failure-mode ablations.  The safe geometric choices do
not separate in the resolved cells; apparent large effects arise from invalid
baselines, unsafe retractions, or insufficient learning-rate grids.}
\label{fig:ablation}
\end{figure*}

\section{How False Positive Advantages Arose}

This section is included because the negative result is only useful if it is
reproducible.  The following four defects each produced a plausible optimizer
advantage in our own workflow before being caught.

\subsection{Trivial benchmark}

Freezing occupied orbitals in a small molecular active-space ladder can leave a
target with zero active electrons.  In the affected LiH construction, several
nominal benchmark entries had exact product-state ground states with Schmidt
rank one and zero entanglement entropy.  Solving such a cell to tiny energy
error is not evidence for a better optimizer; it is evidence that the benchmark
has collapsed to a state-preparation sanity check.

Diagnostic: report active electron count, exact particle number, Schmidt rank,
and entanglement entropy of the exact target before ranking optimizers.

\subsection{Projected matrix Adam}

Adam on unconstrained gate matrix entries followed by a polar projection spends
most of its intended step in the normal direction.  In the measured cells,
$93$--$96\%$ of the step is discarded by projection, and the fraction grows
with system size.  A Riemannian method is expected to beat this strawman, but
the comparison does not establish a practical advantage over a coordinate
baseline that practitioners would actually use.

Diagnostic: decompose the update into tangent and normal components before the
projection and report the normal fraction.

\subsection{Non-unitary retraction}

An inherited Cayley implementation constructed a non-skew-Hermitian generator
when the Adam momentum carried a small normal component.  The resulting update
left the unitary group.  Since the Hamiltonian ground energy is negative, a
state norm above one can push the reported value of $\braopket{\psi}{H}{\psi}$
below $\lambda_{\min}$.  If the metric is the absolute error
$|E-\lambda_{\min}|$, this bug can look like excellent convergence.  In the
10q cell it created an apparent $85\%$ versus $0\%$ success advantage that
disappeared after antisymmetrizing the Cayley generator.

Diagnostic: assert $\|\psi\|=1$, gate unitarity, finite metrics, and
$E\ge\lambda_{\min}$ at every iteration.  A variational algorithm that beats
the variational minimum has a bug.

\subsection{Learning-rate boundary}

In an over-parametrized 8q cell, the first learning-rate grid made the manifold
method appear hundreds of times more accurate than Euler-coordinate Adam.  The
best value for both methods was the smallest learning rate tested.  Extending
the grid downward made both methods reach the exact ground state on all seeds.

Diagnostic: no optimizer comparison is accepted if the best learning rate for
any method lies on the grid boundary.

\begin{table*}[t]
\centering\small
\resizebox{\textwidth}{!}{%
\begin{tabular}{lll}
\toprule
defect & apparent result & diagnostic \\
\midrule
product-state benchmark & tiny VQE error on a trivial target & Schmidt rank, entropy, electron count \\
projected matrix baseline & orders of magnitude geometric advantage & normal fraction of intended step \\
non-unitary retraction & success below variational minimum & state norm, gate unitarity, $E\ge\lambda_{\min}$ \\
learning-rate boundary & method advantage at smallest grid value & expanded grid with interior optimum \\
\bottomrule
\end{tabular}}
\caption{The four false-positive mechanisms.  Each one produced a plausible
multi-seed result before the corresponding diagnostic was enforced.}
\label{tab:false-positive-summary}
\end{table*}

\section{Validation Protocol}

We propose the following minimum checklist for VQE optimizer comparisons.

\begin{enumerate}
  \item Validate that the exact target is non-trivial: Schmidt rank greater
        than one and non-zero entanglement entropy.
  \item Validate that the target is genuinely correlated: the Hartree--Fock
        error should be comfortably above chemical accuracy.
  \item Validate the physical sector, especially particle number for molecular
        Hamiltonians.
  \item Calibrate depth before comparing optimizers.  A cell no tested method
        can solve measures expressivity or budget, not optimizer quality.
  \item Match physical initialization across methods.
  \item Sweep learning rates separately for each method and require interior
        optima for accepted comparisons.
  \item Assert state norm, gate unitarity, finite metrics, and the variational
        lower bound at runtime.
  \item Report both success rate and geometric mean final error with paired
        intervals over shared seeds.
\end{enumerate}

\section{Limitations and Companion Work}

The result is intentionally narrow.  It does not claim that Riemannian
optimization is useless in quantum computing.  It does not address noisy
hardware, measurement-shot effects, large $SU(2^n)$ blocks, Hamiltonian
simulation compilation, or quantum natural gradient with a physically different
metric.  It also does not explain the full depth threshold observed in larger
VQE campaigns.  The companion manuscript studies that mechanism using
architecture-specific representability floors and task-aligned trainability
diagnostics.

\section{Conclusion}

In the tested per-gate $U(2)$ VQE setting, direct Riemannian Adam does not show
a measurable practical advantage over a matched Euler-angle Adam baseline.
Large apparent advantages came from benchmark, baseline, retraction, or tuning
defects.  The main contribution of this first-stage paper is therefore not a new
optimizer but a negative result with a concrete validation protocol for future
optimizer claims.

\begin{IEEEbiographynophoto}{Polina Belozerova}
is an expert in mathematical modeling,
quantitative research, and high-performance computing for modern AI systems.
She graduated from the Department of General and Applied Physics at the Moscow
Institute of Physics and Technology (MIPT) in 2017 and continued her studies
in Computer Science at the Skolkovo Institute of Science and Technology
(Skoltech) in 2019.  Her early research experience includes internships at
Forschungszentrum Juelich and the Caesar Research Centre in Bonn, Germany.
During her Ph.D. studies at MIPT, specializing in modern mathematics, she
received competitive research funding for work in quantum computing and
innovation.  She subsequently led an R\&D direction at Huawei and expanded
her expertise into quantitative finance and AI leadership as Head of Machine
Learning and Chief Product Officer of an AI laboratory at a financial company.
She currently leads the HPC, reinforcement learning, and LLM research
direction at the AI in Dynamic Action (AIDA) Lab, Central University.  Her
research focuses on numerical methods, computational efficiency, and
principled approaches to modern AI for complex dynamic systems.
\end{IEEEbiographynophoto}

\begin{IEEEbiographynophoto}{Pavel Osinenko}
is an expert in advanced control, automation,
and reinforcement learning.  He holds a diploma with honors from Bauman
Moscow State Technical University (2009), a Ph.D. degree from Dresden
University of Technology, Germany (2014), and a habilitation degree from
Chemnitz University of Technology, Germany (2025), officially recognized as
Doctor of Physical and Mathematical Sciences in Russia.  His research
contributions range from foundational aspects of control and dynamic-system
theory to reinforcement learning with guarantees, and include work published
in journals such as \emph{Automatica}, \emph{IEEE Transactions on Automatic
Control}, \emph{IEEE Control Systems Letters}, \emph{IEEE/CAA Journal of
Automatica Sinica}, \emph{Neurocomputing}, and \emph{Journal of Intelligent
Manufacturing}.  He has worked on and supervised R\&D projects involving
robotic and autonomous systems, digital agriculture, automation, and applied
computer vision.  He also developed and led the development of software
frameworks for rapid prototyping and designing reinforcement learning agents,
including Regelum.  He leads the AI in Dynamic Action (AIDA) Lab at Central
University, Moscow, where his central research direction is reinforcement
learning with guarantees combining advanced automatic control with dynamic AI.
\end{IEEEbiographynophoto}

\EOD
\end{document}